\documentclass[11pt]{article}

\usepackage{amsmath,amsthm,amssymb,amsfonts,mathtools}
\usepackage[margin=1in]{geometry}
\usepackage[english]{babel}
\usepackage{xcolor}
\usepackage{hyperref}

\newcommand{\suppsection}[1]{%
  \hyperref[#1]{Supplementary Material, Section~\ref*{#1}}}
\hypersetup{
  colorlinks=true,
  linkcolor=blue,
  citecolor=blue,
  urlcolor=blue,
  filecolor=blue,
  pdftitle={A Frame-Spread Lower Bound for Quantum Entropy Estimation under Fixed Rank-One Measurements},
  pdfauthor={Xinyu Song}
}

\theoremstyle{plain}
\newtheorem{theorem}{Theorem}[section]
\newtheorem{corollary}[theorem]{Corollary}
\newtheorem{proposition}[theorem]{Proposition}
\newtheorem{lemma}[theorem]{Lemma}

\theoremstyle{remark}
\newtheorem{remark}[theorem]{Remark}

\newcommand{\C}{\mathbb C}
\newcommand{\E}{\mathbb E}
\newcommand{\PP}{\mathbb P}
\newcommand{\Hd}{\mathbb H_d}
\newcommand{\Hdz}{\mathbb H_{d,0}}
\newcommand{\dd}{\mathrm{d}}
\newcommand{\tr}{\operatorname{tr}}
\newcommand{\rank}{\operatorname{rank}}
\newcommand{\TV}{\operatorname{TV}}
\newcommand{\norm}[1]{\left\|#1\right\|}

\newcommand{\ket}[1]{\left|#1\right\rangle}
\newcommand{\bra}[1]{\left\langle#1\right|}
\newcommand{\Fr}{\mathcal F_\nu}
\newcommand{\Frz}{\mathcal F_{\nu,0}}

\begin{document}

\title{A Frame-Spread Lower Bound for Quantum Entropy
Estimation under Fixed Rank-One Measurements}

\author{Xinyu Song\\[0.35em]
\small School of Statistics and Data Science\\
\small Shanghai University of Finance and Economics, Shanghai, China\\
\small \texttt{song.xinyu@mail.shufe.edu.cn}}
\date{August 2026}

\maketitle

\begin{abstract}
We study estimation of the von Neumann entropy of a $d$-dimensional
quantum state from independent outcomes of a fixed rank-one measurement.
Writing $\kappa_\nu=(d+1)\|\mathcal F_{\nu,0}\|_{2\to2}$ for the normalized
trace-zero frame spread, one has $1\le\kappa_\nu\le d+1$, with $\kappa_\nu=1$
for exact complex projective $2$-designs. For any fixed rank-one measurement
and any sample size $n\le c_2d^2/\kappa_\nu$, the minimax risk is bounded below
by $c_1\log^2\{d^2/(\kappa_\nu n)\}$. For bounded frame spread and every fixed
$0<a<2$, the minimax risk is of order $\log^2d$ throughout
$n\lesssim d^{2-a}$, so fixed uniform accuracy requires $\Omega(d^2)$
observations. Existing tomography bounds give bounded risk from
$O(d^3\log^3d)$ observations for finite equal-weight exact designs and the
Haar covariant measurement, while the intermediate minimax regime remains
open. The results concern fixed measurements repeated
independently across copies and exclude adaptive and collective protocols and
protocols that reuse common randomness across copies.
\end{abstract}



\section{Introduction}
\label{sec:introduction}

The von Neumann entropy $S(\rho)=-\tr(\rho\log\rho)$ is a spectral
functional of a quantum state $\rho$ on $\C^d$. Existing analyses have
mainly focused on experiments that may process several copies of
the state jointly, using Schur--Weyl duality and empirical Young diagrams
\cite{ODonnellWright2016,AISW2020,WangZhang2025}. Randomized-measurement
and classical-shadow protocols instead draw a fresh random unitary for
each copy, measure that copy separately, and retain a classical record
\cite{HKP2020,ElbenReview2023}. The resulting records are independent
outcomes from one fixed rank-one positive operator-valued measure (POVM).
For a review of quantum measurement and state estimation from a statistical
perspective, see \cite{WangSong2020}.

This paper studies the number of independent outcomes required for
uniformly accurate entropy estimation in this fixed-measurement
experiment. For a fixed rank-one POVM $\nu$, let $\Fr$ be its
second-moment operator on Hermitian matrices, and let $\kappa_\nu$ be
$d+1$ times the largest eigenvalue of $\Fr$ on trace-zero matrices. Then
$1\le\kappa_\nu\le d+1$. Thus $\kappa_\nu$ measures the largest weighting of
a trace-zero direction relative to the isotropic $2$-design frame; a smaller
value yields a stronger lower bound. Moreover, $\kappa_\nu=1$ when $\nu$ is a
complex projective $2$-design. This class contains the Haar covariant
measurement in every dimension and, when $d=2^q$, the uniform stabilizer
measurement generated by the global Clifford protocol \cite{Zhu2017}.

The main result is the minimax mean-squared-error lower bound
\[
  R_{n,d}^*(\nu)\ge c_1\log^2\!\left(\frac{d^2}{\kappa_\nu n}\right),
  \qquad 1\le n\le\frac{c_2d^2}{\kappa_\nu},
\]
where $c_1,c_2>0$ are universal constants. For measurements with bounded
$\kappa_\nu$, the lower bound follows the curve $\log^2(d^2/n)$. For every
fixed $0<a<2$, the minimax risk throughout
$n\lesssim d^{2-a}$ is of the same order $\log^2d$ as the risk of a
constant estimator. The same bounds hold for measurements with bounded
relative frame error. Existing projected least-squares tomography
bounds, combined with entropy continuity, give uniformly bounded entropy
risk for finite exact designs equipped with equal weights and for the Haar
covariant measurement when $n\ge Cd^3\log^3d$. For these measurements,
and for any sufficiently small fixed target mean squared error, the required
sample size lies between order $d^2$ and order $d^3\log^3d$. The minimax
behavior in the intermediate regime remains open.

For fixed deterministic unentangled measurements, Liu and Acharya
\cite{LiuAcharya2024} proved that the optimal copy complexity of mixedness
testing is $\Theta(d^2/\varepsilon^2)$. Their L\"uders channel coincides with
our frame operator for rank-one POVMs, and quantum Pinsker plus thresholding
already gives an $\Omega(d^2)$ requirement at a sufficiently small fixed
entropy risk. That bounded-distance reduction yields only a fixed positive
risk. Theorem~\ref{thm:main} instead gives a frame-dependent curve that grows
with $d$ and, with the constant estimator, is sharp up to constants throughout
the polynomially subquadratic regimes of Corollary~\ref{cor:two-sided}.

Recent preprints report lower bounds for the unrestricted quantum experiment,
which permits collective measurements on all copies
\cite{FOW2026,Wang2026}. Gao and Wang \cite{GaoWang2026} construct a
subquadratic estimator in the same model. The lower bounds also apply to the
fixed-measurement experiment but do not provide the frame-dependent risk curve.
After clipping the estimator of \cite{GaoWang2026} and applying
$O(\log\log d)$-fold median amplification, one obtains bounded mean squared
error with $O\{d^2(\log\log d)^3/\log^2d\}$ copies. For sample sizes of this
order and all sufficiently large $d$, Theorem~\ref{thm:main} gives risk of
order at least $\log^2\log d$ for bounded frame spread.
Section~\ref{sec:discussion} details this growing-risk separation.

Our proofs compare the maximally mixed state with a Haar mixture of flat
rank-$r$ states. The mixture chi-square divergence reduces through $\Fr$ to a
kernel of two prior draws. Conditioning on one draw and applying unitary
concentration to the other bounds its exponential moment in terms of
$\kappa_\nu$, without a design identity. Related mixture arguments appear in
single-copy tomography and testing
\cite{ChenHuangLiLiu2022,LoweNayak2025}.

Section~\ref{sec:model} defines the experiment and frame spread;
Section~\ref{sec:results} gives the results; and Section~\ref{sec:discussion}
gives comparisons and scope. Complete proofs are in the Supplementary Material.

\section{Fixed rank-one measurements and frame spread}
\label{sec:model}

Let $\Theta_d$ be the set of density matrices on $\C^d$, let $\Hd$ be
the real Hilbert space of $d\times d$ Hermitian matrices with inner
product $\langle A,B\rangle=\tr(AB)$ and norm $\norm{\cdot}_F$, and let
$\Hdz$ be its trace-zero subspace. Let $(\mathcal Z,\mathcal A,\nu)$ be
a probability space and $z\mapsto P_z=\ket z\bra z$ a measurable map
into rank-one projectors such that
\begin{equation}
  \int P_z\,\nu(\dd z)=\frac Id .
  \label{eq:normalization}
\end{equation}
Then $M(\dd z)=dP_z\,\nu(\dd z)$ is a rank-one POVM, and one observation
$Z$ under the state $\rho$ has density $p_\rho(z)=d\tr(\rho P_z)$ with
respect to $\nu$. Under $\rho_0=I/d$ the density is identically one. The
same measurement is applied independently to each of $n$ copies, so the
data are independent and identically distributed. In the global
Clifford protocol a uniform Clifford unitary is drawn, the copy is
measured in the computational basis, and both are recorded; the recorded
direction is a uniform stabilizer state, and the pair is statistically
equivalent to it (\suppsection{supp-sec:models}). The Haar covariant
measurement is the case in which $\nu$ is the Haar measure on pure
states.

Define the frame operator
\[
  \Fr(A)=d\int\tr(AP_z)\,P_z\,\nu(\dd z),\qquad A\in\Hd .
\]
By \eqref{eq:normalization}, $\Fr(I)=I$ and $\tr\Fr(A)=\tr A$, so
$\Fr$ leaves $\Hdz$ invariant; let $\Frz$ denote the restriction. The
\emph{frame spread} of $\nu$ is
\[
  \kappa_\nu=(d+1)\,\norm{\Frz}_{2\to2},
\]
the operator norm being taken with respect to $\norm{\cdot}_F$.

\begin{lemma}
\label{lem:frame}
Let $d\ge2$ and let $\nu$ satisfy \eqref{eq:normalization}.
\begin{enumerate}
\item[(i)] $\Fr$ is self-adjoint and positive semidefinite on $\Hd$, and
$1\le\kappa_\nu\le d+1$.
\item[(ii)] If $\nu$ is an exact complex projective $2$-design, that is,
$\int P_z^{\otimes2}\,\nu(\dd z)=(I+\mathsf F)/\{d(d+1)\}$ with
$\mathsf F$ the swap operator, then $\Frz=(d+1)^{-1}I$ and
$\kappa_\nu=1$.
\item[(iii)] If $\norm{(d+1)\Frz-I}_{2\to2}\le\eta$, then
$\kappa_\nu\le1+\eta$.
\item[(iv)] For all $\rho,\sigma\in\Theta_d$,
\begin{equation}
  \int p_\rho(z)p_\sigma(z)\,\nu(\dd z)
  =1+d\,\tr\!\left\{(\rho-I/d)\,\Frz(\sigma-I/d)\right\}.
  \label{eq:affinity}
\end{equation}
\end{enumerate}
\end{lemma}

For the lower-bound argument, the measurement enters through the affinity
identity \eqref{eq:affinity}. For a prior $\pi$ on $\Theta_d$, let
$P_\pi^{(n)}=\int P_\rho^{(n)}\,\pi(\dd\rho)$ be the mixture of the
$n$-sample laws. Independence and Fubini's theorem give the mixture
chi-square identity
\begin{equation}
  1+\chi^2\!\left(P_\pi^{(n)}\,\middle\|\,P_{I/d}^{(n)}\right)
  =\E_{\rho,\rho'\stackrel{\mathrm{iid}}{\sim}\pi}
  \left[\left\{1+d\,\tr\!\left((\rho-I/d)\,\Frz(\rho'-I/d)\right)
  \right\}^{n}\right].
  \label{eq:mixture-chi2}
\end{equation}
Under an exact $2$-design the bracket equals $d\{1+\tr(\rho\rho')\}/(d+1)$.

Let $\mathcal E_n$ be the class of measurable maps from $\mathcal Z^n$
to $\mathbb R$. For $\widehat S\in\mathcal E_n$, define
\[
  R_{n,d}(\widehat S;\nu)
  =\sup_{\rho\in\Theta_d}\E_\rho\{\widehat S-S(\rho)\}^2,
  \qquad
  R_{n,d}^*(\nu)=\inf_{\widehat S\in\mathcal E_n}R_{n,d}(\widehat S;\nu).
\]

\section{Main results}
\label{sec:results}

\begin{theorem}
\label{thm:main}
For every $d\ge2$, every rank-one POVM $\nu$ satisfying
\eqref{eq:normalization}, and every integer $n$ with
$1\le n\le d^2/(150\kappa_\nu)$,
\[
  R_{n,d}^*(\nu)
  \ge\frac1{72}\log^2\!\left(\frac{d^2}{\kappa_\nu n}\right).
\]
\end{theorem}

The numerical constants $c_1=1/72$ and $c_2=1/150$ are chosen for a
simple finite-sample statement and are not asserted to be sharp. The stated
sample range is nonempty only when $d^2\ge150\kappa_\nu$. The
theorem imposes no informational-completeness condition. For a
noninformationally complete measurement, the bound remains valid but
need not be sharp; see Section~\ref{sec:discussion}.

\begin{corollary}
\label{cor:bounded}
Let $K\ge1$ and let $\nu$ satisfy $\kappa_\nu\le K$.
\begin{enumerate}
\item[(i)] For $1\le n\le d^2/(150K)$,
$R_{n,d}^*(\nu)\ge\frac1{72}\log^2\{d^2/(Kn)\}$.
\item[(ii)] For every $\varepsilon>0$,
\[
  R_{n,d}^*(\nu)\le\varepsilon
  \quad\Longrightarrow\quad
  n\ge\frac{d^2}K\min\left\{\frac1{150},\,e^{-\sqrt{72\varepsilon}}\right\}.
\]
\end{enumerate}
Both statements hold with $K=1$ for every exact complex projective
$2$-design, in particular for the Haar covariant measurement in every
dimension and for the global Clifford measurement when $d=2^q$, and with
$K=1+\eta$ for every $\nu$ with
$\norm{(d+1)\Frz-I}_{2\to2}\le\eta$.
\end{corollary}

\begin{corollary}
\label{cor:two-sided}
Let $\kappa_\nu\le K$, let $0<a<2$, and let
$1\le n\le d^{2-a}/(150K)$. Then
\[
  \frac{a^2}{72}\log^2d
  \le R_{n,d}^*(\nu)
  \le\frac14\log^2d .
\]
\end{corollary}

The constant estimator $\widehat S\equiv\frac12\log d$ attains the upper
bound in Corollary~\ref{cor:two-sided}. Thus, in every polynomially
subquadratic regime covered by the corollary, the minimax risk is
comparable to the squared diameter of the entropy range. The corollary
makes no assertion about the transition region $n\asymp d^2$.

\begin{proposition}
\label{prop:upper}
There are universal constants $C,C'>0$ such that the following holds
for every $d\ge2$, every finite exact complex projective $2$-design with
uniform design measure (equal weights on the design points), and the Haar
covariant measurement. If
$n\ge Cd^3\log d$, the plug-in
entropy $\widehat S_n=S(\widehat\rho_n)$ of the projected least-squares
estimator $\widehat\rho_n$ of \cite{GKKT2020} satisfies
\[
  \sup_{\rho\in\Theta_d}\E_\rho\{\widehat S_n-S(\rho)\}^2
  \le C'\left\{\frac{d^3\log^3d}{n}
  +\left(\frac{d^3\log d}{n}\right)^{1/2}\right\}.
\]
\end{proposition}

Proposition~\ref{prop:upper} combines existing tomography and
entropy-continuity bounds to give a nonmatching achievable regime. For the
measurements covered by the proposition, let
$N_\tau(d,\nu)=\inf\{n\in\mathbb N:R_{n,d}^*(\nu)\le\tau\}$ for $\tau>0$.
For every fixed $0<\tau<\frac1{72}\log^2(150)$, Theorem~\ref{thm:main} and
Proposition~\ref{prop:upper} imply
\[
  \left\lfloor\frac{d^2}{150}\right\rfloor<N_\tau(d,\nu)
  \le C_\tau d^3\log^3d,
\]
where $C_\tau$ depends only on $\tau$. The minimax rate in the intermediate
regime remains open.

The proof of Theorem~\ref{thm:main} proceeds by comparing the maximally
mixed state with a Haar mixture of flat rank-$r$ states. Fix a rank-$r$
projector $\Pi_r$ with $r\le d$ chosen below, and let
$\pi_r$ be the law of $\rho_U=U\Pi_rU^*/r$ for $U$ Haar distributed on
the unitary group $\mathrm U(d)$. Every state in the support of $\pi_r$
has entropy $\log r$, and the null state $I/d$ has entropy $\log d$, so
the two hypotheses are separated by $\delta=\log(d/r)$. Write
$A_U=\rho_U-I/d$, so that $\norm{A_U}_F^2=1/r-1/d$, and define the
kernel $K_\nu(U,V)=d\tr\{A_U\Frz(A_V)\}$ for two independent prior
draws. By \eqref{eq:mixture-chi2},
\[
  1+\chi^2\!\left(P_{\pi_r}^{(n)}\,\middle\|\,P_{I/d}^{(n)}\right)
  =\E_{U,V}\{1+K_\nu(U,V)\}^n\le\E_V\E_U\exp\{nK_\nu(U,V)\}.
\]
Conditionally on $V$, the map $U\mapsto K_\nu(U,V)=\tr(U\Lambda U^*B_V)$,
with $\Lambda=\Pi_r/r$ and $B_V=d\Frz(A_V)$, has mean zero and is
Lipschitz with constant at most
$2\norm\Lambda_{\mathrm{op}}\norm{B_V}_F\le2\kappa_\nu r^{-3/2}$ in the
Frobenius metric. The logarithmic Sobolev inequality for Haar measure on
$\mathrm U(d)$ \cite[Theorem~5.16]{Meckes2019} then gives
\[
  \E_U\exp\{nK_\nu(U,V)\}
  \le\exp\!\left(\frac{12\kappa_\nu^2n^2}{dr^3}\right)
\]
uniformly in $V$. Taking $r=\lceil\max\{1,(18\kappa_\nu^2n^2/d)^{1/3}\}
\rceil$ makes the right side at most $e^{2/3}<2$, so the chi-square
divergence is at most one and the total variation distance at most
$1/2$. On the range $n\le d^2/(150\kappa_\nu)$ this choice satisfies
$r\le d$ and $\delta\ge\frac13\log\{d^2/(\kappa_\nu n)\}$. Le Cam's
two-point inequality gives $R_{n,d}^*(\nu)\ge\delta^2/8$. Details,
including the corollaries and Proposition~\ref{prop:upper}, are in
\suppsection{supp-sec:proof} and \suppsection{supp-sec:upper}.

\section{Discussion}
\label{sec:discussion}

\paragraph{Relation to single-copy testing bounds}
Mixedness testing distinguishes $I/d$ from states at trace distance at least
$\varepsilon$. Liu and Acharya \cite{LiuAcharya2024} proved that its optimal
copy complexity under fixed deterministic unentangled measurements is
$\Theta(d^2/\varepsilon^2)$. Their protocol may use a predetermined sequence
of POVMs, so one fixed repeated POVM is a special case. For a finite rank-one POVM
$M_x=w_xP_x$, its L\"uders channel is
$\mathcal H_M(A)=\sum_xw_x\tr(AP_x)P_x$, which is $\Fr(A)$ under the
normalization in Section~\ref{sec:model}. Thus the channel appearing in
\cite{LiuAcharya2024} is the frame operator used in the present analysis.
The quantum Pinsker inequality gives
$\log d-S(\rho)\ge\frac12\norm{\rho-I/d}_1^2$. If an entropy estimator has
uniform mean squared error at most $\varepsilon^4/48$, thresholding it at
$\log d-\varepsilon^2/4$ gives a mixedness tester with both error
probabilities at most $1/3$. Thus, for a fixed $\varepsilon$ in the range of
the testing theorem, the fixed-risk quadratic scale already follows from
\cite{LiuAcharya2024}. This bounded-distance reduction cannot produce a risk
that diverges with $d$. Theorem~\ref{thm:main} additionally combines the
largest trace-zero eigenvalue of the channel with flat rank-$r$ alternatives,
whose entropy gap is $\log(d/r)$, yielding the curve
$\log^2\{d^2/(\kappa_\nu n)\}$. Randomized or adaptive single-copy testing has
complexity of order $d^{3/2}/\varepsilon^2$
\cite{BCL2020,ChenHuangLiLiu2022}; the common random basis in
\cite[Appendix~A]{BCL2020} lies outside the repeated fixed-POVM model. The
bounds of \cite{LoweNayak2025} concern tomography and shadow tomography.

\paragraph{Scope}
The theorem concerns one fixed rank-one POVM applied independently to
each copy. Adaptive protocols, protocols that reuse common randomness, and
collective measurements do not generate the repeated product law assumed by
the theorem. When $\kappa_\nu$ is of order $d$ the bound can be conservative:
for one orthonormal basis, $\kappa_\nu=d+1$, yet $I/d$ and a pure state that is
unbiased with respect to that basis have the same law, giving minimax risk at
least
$\frac14\log^2d$ for every $n$. Finite higher-rank POVMs are covered through
rank-one refinement; see
\suppsection{supp-sec:scope}.

\paragraph{Relation to unrestricted experiments}
For collective measurements, \cite{AISW2020} proved a dimension-quadratic
lower bound for the empirical Young diagram estimator and left the
unrestricted problem open. Fanizza et al.~\cite{FOW2026} report a
$d^{2-\gamma}$ lower bound for every $\gamma>0$ at constant error, and Wang
\cite{Wang2026} reports an $\varepsilon$-dependent bound whose
constant-accuracy specialization is $d^2/\log^4d$. These bounds also apply to
fixed measurements. Theorem~\ref{thm:main} adds the frame-dependent scale
$d^2/\kappa_\nu$ without a dimension-dependent polylogarithmic loss and its
finite-sample risk curve. Gao and Wang \cite{GaoWang2026} give an unrestricted
constant-accuracy estimator with success probability at least $2/3$. For a
mean-squared-error comparison, clipping its output to $[0,\log d]$ and taking
the median of $O(\log\log d)$ independent runs reduces the failure probability
to
$O(1/\log^2d)$, giving bounded mean squared error with
$O\{d^2(\log\log d)^3/\log^2d\}$ copies. For sample sizes of this order and
all sufficiently large $d$, Theorem~\ref{thm:main} gives risk of order at least
$\log^2\log d$ for every bounded-frame-spread measurement. This
growing-risk separation pertains to the fixed-measurement model.

\clearpage
\setcounter{section}{0}
\setcounter{equation}{0}
\renewcommand{\thesection}{S\arabic{section}}
\renewcommand{\theequation}{S\arabic{equation}}
\renewcommand{\theHsection}{supp.\arabic{section}}
\renewcommand{\theHequation}{supp.\arabic{equation}}

\begin{center}
{\Large\bfseries Supplementary Material}\par
\vspace{0.5em}
{\large A Frame-Spread Lower Bound for Quantum Entropy Estimation
under Fixed Rank-One Measurements}\par
\vspace{0.5em}
Xinyu Song
\end{center}

\medskip
Throughout, $\nu$ is a probability measure on rank-one projectors
satisfying $\int P_z\,\nu(\dd z)=I/d$, $p_\rho(z)=d\tr(\rho P_z)$ is the
one-observation density with respect to $\nu$, $\Fr$ is the frame
operator, $\Frz$ its restriction to the trace-zero subspace $\Hdz$, and
$\kappa_\nu=(d+1)\norm{\Frz}_{2\to2}$, as in Section~2 of the main
article. Logarithms are natural.

\section{Randomized bases and rank-one POVMs}
\label{supp-sec:models}

This section shows that the randomized-basis record and the rank-one POVM
observation used in the main article define statistically equivalent
experiments. Let $\mu$ be a probability
measure on unitary matrices, and define
\[
  P_{U,b}=U^*\ket b\bra bU,\qquad b\in\{1,\ldots,d\}.
\]
In the randomized-basis experiment, $U\sim\mu$ is drawn, a fresh copy
of $\rho$ is measured in the basis determined by $U$, and the pair
$(U,b)$ is recorded. Its density with respect to the product of
$\mu(\dd U)$ and counting measure on $b$ is
$f_\rho(U,b)=\tr(\rho P_{U,b})$. Let $\nu_\mu$ be the law of $P_{U,B}$
when $U\sim\mu$ and $B$ is independent and uniform on
$\{1,\ldots,d\}$. For any bounded measurable $h$,
\[
  \int h(P)\,\nu_\mu(\dd P)
  =\frac1d\int\sum_{b=1}^d h(P_{U,b})\,\mu(\dd U).
\]
The induced projector $P=P_{U,b}$ has density $p_\rho(P)=d\tr(\rho
P)$ with respect to $\nu_\mu$. Conditional on $P$, the distribution on
the fiber $\{(U,b):P_{U,b}=P\}$ is independent of $\rho$, because
$f_\rho(U,b)$ depends on $(U,b)$ only through $P_{U,b}$. The direction
$P$ is therefore sufficient, and the two experiments are statistically
equivalent. Since $\sum_b P_{U,b}=I$ for every $U$, the measure
$\nu_\mu$ satisfies the normalization $\int P\,\nu_\mu(\dd P)=I/d$.

If $\mu$ is Haar, then $\nu_\mu$ is Haar measure on pure states. If
$\mu$ is uniform on the multiqubit Clifford group, then $\nu_\mu$ is
uniform on stabilizer states. Stabilizer states form an exact complex
projective $3$-design \cite{Zhu2017} and hence an exact projective
$2$-design.

\section{Frame operator and affinity identities}
\label{supp-sec:frame}

\begin{proof}[Proof of Lemma~\ref{lem:frame} in the main article]
(i) For $A,B\in\Hd$,
\[
  \langle B,\Fr(A)\rangle
  =d\int\tr(AP_z)\tr(BP_z)\,\nu(\dd z),
\]
which is symmetric in $(A,B)$ and nonnegative when $A=B$. Hence $\Fr$
is self-adjoint and positive semidefinite. For the upper bound, write
$P_z=\ket z\bra z$. Since $\langle z|A|z\rangle^2\le\langle z|A^2|z
\rangle$ by the Cauchy--Schwarz inequality,
\[
  \langle A,\Fr(A)\rangle
  =d\int\langle z|A|z\rangle^2\,\nu(\dd z)
  \le d\int\langle z|A^2|z\rangle\,\nu(\dd z)
  =\tr\!\left(A^2\,d\!\int P_z\,\nu(\dd z)\right)=\tr(A^2).
\]
Thus $\norm{\Frz}_{2\to2}\le1$ and $\kappa_\nu\le d+1$. For the lower
bound, let $\{E_k\}_{k=1}^{d^2}$ be an orthonormal basis of $\Hd$ with
$E_1=I/\sqrt d$. By Parseval's identity,
\[
  \sum_{k=1}^{d^2}\langle E_k,\Fr(E_k)\rangle
  =d\int\sum_k\tr(E_kP_z)^2\,\nu(\dd z)
  =d\int\norm{P_z}_F^2\,\nu(\dd z)=d .
\]
The term $k=1$ equals $\tr(I\cdot I)/d=1$, so the trace of $\Frz$ on
the $(d^2-1)$-dimensional space $\Hdz$ equals $d-1$. Its largest
eigenvalue is at least the average eigenvalue $(d-1)/(d^2-1)=1/(d+1)$,
which gives $\kappa_\nu\ge1$.

(ii) Taking the partial trace over the first tensor factor in the
$2$-design identity gives, for every $A\in\Hd$,
\[
  \Fr(A)=d\,\tr_1\!\left\{(A\otimes I)\int P_z^{\otimes2}\,\nu(\dd z)
  \right\}
  =\frac{\tr_1\{(A\otimes I)(I+\mathsf F)\}}{d+1}
  =\frac{\tr(A)\,I+A}{d+1},
\]
because $\tr_1\{(A\otimes I)\mathsf F\}=A$. On $\Hdz$ this is
$A/(d+1)$, so $\Frz=(d+1)^{-1}I$ and $\kappa_\nu=1$.

(iii) By the triangle inequality,
$(d+1)\norm{\Frz}_{2\to2}\le\norm{I}_{2\to2}+\eta=1+\eta$.

(iv) Since $\int p_\rho p_\sigma\,\dd\nu=d^2\int\tr(\rho P_z)\tr(\sigma
P_z)\,\nu(\dd z)=d\,\tr\{\rho\Fr(\sigma)\}$, write $\rho=I/d+A$ and
$\sigma=I/d+B$ with $A,B\in\Hdz$. Then $\Fr(\sigma)=I/d+\Frz(B)$ and,
using $\tr A=0$ and $\tr\Frz(B)=\tr B=0$,
\[
  \tr\{\rho\Fr(\sigma)\}
  =\frac1d+\frac{\tr\Frz(B)}d+\frac{\tr A}d+\tr\{A\Frz(B)\}
  =\frac1d+\tr\{A\Frz(B)\}.
\]
Multiplying by $d$ proves \eqref{eq:affinity}.
\end{proof}

\begin{lemma}[Mixture chi-square identity]
\label{supp-lem:chi2}
Let $P_\rho^{(n)}$ be the law of $n$ independent observations with
density $p_\rho$, and let $\pi$ be any prior on density matrices. Then
\eqref{eq:mixture-chi2} holds.
\end{lemma}

\begin{proof}
Under $\rho_0=I/d$, $p_{\rho_0}\equiv1$. Expanding the squared mixture
likelihood ratio and applying Fubini's theorem,
\begin{align*}
  1+\chi^2\!\left(P_\pi^{(n)}\,\middle\|\,P_{\rho_0}^{(n)}\right)
  &=\int\left\{\int\prod_{i=1}^n p_\rho(z_i)\,\pi(\dd\rho)\right\}^2
  \nu^{\otimes n}(\dd z_1\cdots\dd z_n)\\
  &=\E_{\rho,\rho'\stackrel{\mathrm{iid}}{\sim}\pi}
  \prod_{i=1}^n\int p_\rho(z_i)p_{\rho'}(z_i)\,\nu(\dd z_i),
\end{align*}
and each factor equals
$1+d\tr\{(\rho-I/d)\Frz(\rho'-I/d)\}$ by
\eqref{eq:affinity}.
\end{proof}

\section{Proof of Theorem~\ref{thm:main} and its corollaries}
\label{supp-sec:proof}

\begin{lemma}[Unitary concentration]
\label{supp-lem:concentration}
Let $d\ge1$, let $U$ be Haar distributed on $\mathrm U(d)$, and let
$f:\mathrm U(d)\to\mathbb R$ be $L$-Lipschitz with respect to the
Frobenius (Hilbert--Schmidt) metric. Then, for every real $\lambda$,
\[
  \E\exp\{\lambda(f(U)-\E f(U))\}
  \le\exp\!\left(\frac{3\lambda^2L^2}{d}\right).
\]
\end{lemma}

\begin{proof}
By \cite[Theorem~5.16]{Meckes2019}, Haar measure on $\mathrm U(d)$
equipped with the Hilbert--Schmidt metric satisfies a logarithmic
Sobolev inequality with constant $C=6/d$, in the normalization
$\operatorname{Ent}(g^2)\le2C\,\E|\nabla g|^2$ used there. The Herbst
argument in the proof of \cite[Theorem~5.5]{Meckes2019} converts this
inequality into the bound $\E\exp\{\lambda(F-\E F)\}\le\exp(C\lambda^2/2)$
for every $1$-Lipschitz $F$ and real $\lambda$. Applying it to
$F=f/L$ gives the claim.
\end{proof}

\begin{remark}
Theorem~5.17 of \cite{Meckes2019} gives the corresponding tail inequality
$\PP\{f(U)\ge\E f(U)+t\}\le\exp\{-(d-2)t^2/(24L^2)\}$, with the factor
$d-2$ representing the worst case among the compact groups covered there.
For $\mathrm U(d)$, Theorem~5.16 gives logarithmic Sobolev constant $6/d$.
Lemma~5.15 treats the central phase direction by coupling Haar measures on
$\mathrm U(d)$ and $\mathrm{SU}(d)$, so the concentration inequality does
not require phase invariance. The kernel used below is also invariant
under $U\mapsto e^{i\theta}U$.
\end{remark}

Fix an integer $1\le r\le d$ and a rank-$r$ projector $\Pi_r$. For
$U\in\mathrm U(d)$ let
\[
  \rho_U=\frac1rU\Pi_rU^*,\qquad A_U=\rho_U-\frac Id,\qquad
  \norm{A_U}_F^2=\tr(\rho_U^2)-\frac1d=\frac1r-\frac1d\le\frac1r .
\]
Let $\pi_r$ be the law of $\rho_U$ for Haar distributed $U$, and for
independent Haar $U,V$ define
\[
  K_\nu(U,V)=d\,\tr\{A_U\Frz(A_V)\}.
\]

\begin{lemma}[Conditional concentration of the kernel]
\label{supp-lem:kernel}
For every $V\in\mathrm U(d)$ and every real $t$,
\[
  \E_U\exp\{tK_\nu(U,V)\}
  \le\exp\!\left(\frac{12\kappa_\nu^2t^2}{dr^3}\right).
\]
\end{lemma}

\begin{proof}
Put $\Lambda=\Pi_r/r$ and $B_V=d\Frz(A_V)\in\Hdz$. Since $\tr B_V=0$,
\[
  K_\nu(U,V)=\tr\{(U\Lambda U^*-I/d)B_V\}=\tr(U\Lambda U^*B_V),
\]
and $\E_U(U\Lambda U^*)=\tr(\Lambda)I/d=I/d$ gives $\E_UK_\nu(U,V)=0$.
Moreover,
\[
  \norm{B_V}_F
  \le d\norm{\Frz}_{2\to2}\norm{A_V}_F
  =\frac{d\kappa_\nu}{d+1}\norm{A_V}_F
  \le\frac{\kappa_\nu}{\sqrt r}.
\]
For $U,U'\in\mathrm U(d)$ and $D=U-U'$,
\[
  \tr(U\Lambda U^*B_V)-\tr(U'\Lambda U'^*B_V)
  =\tr(D\Lambda U^*B_V)+\tr(U'\Lambda D^*B_V),
\]
and each term is bounded in absolute value by
$\norm D_F\norm{\Lambda}_{\mathrm{op}}\norm{B_V}_F$, because
$|\tr(D\Lambda U^*B_V)|\le\norm D_F\norm{\Lambda U^*B_V}_F\le\norm D_F
\norm\Lambda_{\mathrm{op}}\norm{U^*B_V}_F$ and $\norm{U^*B_V}_F=\norm{B_V}_F$.
Hence $U\mapsto K_\nu(U,V)$ is $L$-Lipschitz with
\[
  L\le2\norm\Lambda_{\mathrm{op}}\norm{B_V}_F
  \le\frac{2\kappa_\nu}{r^{3/2}} .
\]
Lemma~\ref{supp-lem:concentration} gives
$\E_U\exp\{tK_\nu(U,V)\}\le\exp(3t^2L^2/d)
\le\exp\{12\kappa_\nu^2t^2/(dr^3)\}$.
\end{proof}

\begin{proof}[Proof of Theorem~\ref{thm:main}]
Let $d\ge2$ and $1\le n\le d^2/(150\kappa_\nu)$, and write
\[
  y=\frac{d^2}{\kappa_\nu n}\ge150,\qquad
  a=\left(\frac{18\kappa_\nu^2n^2}{d}\right)^{1/3}
  =18^{1/3}\,d\,y^{-2/3},\qquad
  r=\lceil\max\{1,a\}\rceil .
\]

\emph{Step 1: the rank is admissible.} Since $y\ge150$,
$a\le18^{1/3}150^{-2/3}d<0.1d$. If $a\ge1$ then $r\le a+1\le2a<d$; if
$a<1$ then $r=1\le d$. In both cases $1\le r\le d$.

\emph{Step 2: chi-square control.} Consider the null $\rho_0=I/d$ and
the mixture alternative $\rho\sim\pi_r$. By Lemma~\ref{supp-lem:chi2},
$1+\chi^2(P_{\pi_r}^{(n)}\|P_{I/d}^{(n)})=\E_{U,V}\{1+K_\nu(U,V)\}^n$.
The quantity $1+K_\nu(U,V)$ is an integral of a product of two
nonnegative densities, hence nonnegative, so $(1+K)^n\le e^{nK}$.
Conditioning on $V$ and applying Lemma~\ref{supp-lem:kernel} with
$t=n$,
\[
  1+\chi^2\!\left(P_{\pi_r}^{(n)}\,\middle\|\,P_{I/d}^{(n)}\right)
  \le\exp\!\left(\frac{12\kappa_\nu^2n^2}{dr^3}\right).
\]
If $a\ge1$ then $r\ge a$ and $r^3\ge18\kappa_\nu^2n^2/d$, so the
exponent is at most $12/18=2/3$. If $a<1$ then $r=1$ and the exponent
equals $(12/18)a^3<2/3$. Hence $1+\chi^2\le e^{2/3}<2$, and
\[
  \TV\!\left(P_{\pi_r}^{(n)},P_{I/d}^{(n)}\right)
  \le\frac12\sqrt{\chi^2}<\frac12 .
\]

\emph{Step 3: entropy separation.} Every state in the support of
$\pi_r$ has entropy $\log r$ and the null state has entropy $\log d$;
let $\delta=\log(d/r)$. If $a\ge1$, Step~1 gives
$r\le2a=2\cdot18^{1/3}\,d\,y^{-2/3}$, so
\[
  \delta\ge\frac23\log y-\log(2\cdot18^{1/3})\ge\frac13\log y,
\]
because $\log y\ge\log150>4.97\ge3\log(2\cdot18^{1/3})$. If $a<1$ then
$r=1$ and $\delta=\log d$; since $\kappa_\nu n\ge1$ we have $y\le d^2$,
so $\delta\ge\frac12\log y\ge\frac13\log y$. In both cases
$\delta\ge\frac13\log\{d^2/(\kappa_\nu n)\}$.

\emph{Step 4: Le Cam reduction.} For any estimator $\widehat S$, let
$\E_0$ and $\E_1$ denote expectation under the null law and under the
mixture law. Since $|\log d-\log r|=\delta$, the elementary inequality
$x^2+(\delta-x)^2\ge\delta^2/2$ gives
$(\widehat S-\log d)^2+(\widehat S-\log r)^2\ge\delta^2/2$ at every data
point, so
\begin{align*}
  \E_0(\widehat S-\log d)^2+\E_1(\widehat S-\log r)^2
  &\ge\frac{\delta^2}{2}\int\min\{\dd P_{I/d}^{(n)},\dd P_{\pi_r}^{(n)}\}
  =\frac{\delta^2}{2}\left\{1-\TV(P_{I/d}^{(n)},P_{\pi_r}^{(n)})\right\}
  \ge\frac{\delta^2}{4}.
\end{align*}
Hence one of the two risks is at least $\delta^2/8$. Under the
alternative, the mixture risk is an average over states in the support
of $\pi_r$, all with entropy $\log r$, and is therefore bounded by the
worst-case risk. Taking the infimum over estimators and using Step~3,
\[
  R_{n,d}^*(\nu)\ge\frac{\delta^2}{8}
  \ge\frac1{72}\log^2\!\left(\frac{d^2}{\kappa_\nu n}\right).
\]
This is the claim with $c_1=1/72$.
\end{proof}

\begin{proof}[Proof of Corollary~\ref{cor:bounded}]
(i) If $n\le d^2/(150K)$ then $n\le d^2/(150\kappa_\nu)$, so
Theorem~\ref{thm:main} applies, and
$d^2/(\kappa_\nu n)\ge d^2/(Kn)\ge150>1$ gives
$\log^2\{d^2/(\kappa_\nu n)\}\ge\log^2\{d^2/(Kn)\}$.

(ii) Suppose $R_{n,d}^*(\nu)\le\varepsilon$. If $n\le d^2/(150K)$, part
(i) gives $\log\{d^2/(Kn)\}\le\sqrt{72\varepsilon}$, hence
$n\ge(d^2/K)e^{-\sqrt{72\varepsilon}}$. Otherwise $n>d^2/(150K)$. In
either case $n\ge(d^2/K)\min\{1/150,e^{-\sqrt{72\varepsilon}}\}$.

The statements for exact $2$-designs follow from
Lemma~\ref{lem:frame}(ii), for stabilizer states from
Section~\ref{supp-sec:models}, and for approximate designs from
Lemma~\ref{lem:frame}(iii).
\end{proof}

\begin{proof}[Proof of Corollary~\ref{cor:two-sided}]
If $n\le d^{2-a}/(150K)$ then $n\le d^2/(150\kappa_\nu)$ and
$d^2/(\kappa_\nu n)\ge150Kd^a/\kappa_\nu\ge d^a$, so
Theorem~\ref{thm:main} gives the lower bound
$\frac1{72}(a\log d)^2$. For the upper bound, $S(\rho)\in[0,\log d]$
for every $\rho\in\Theta_d$, so the constant estimator
$\widehat S\equiv\frac12\log d$ has squared error at most
$\frac14\log^2d$ at every state.
\end{proof}

\section{Proof of Proposition~\ref{prop:upper}}
\label{supp-sec:upper}

Let $\widehat\rho_n$ be the projected least-squares estimator of
\cite{GKKT2020}. It is obtained by projecting the linear-inversion estimate
onto $\Theta_d$ in Frobenius norm. For a finite exact $2$-design with uniform
design measure, called an equal-weight ``structured POVM'' with $g(d)=2d$ in
\cite{GKKT2020}, Theorem~1 of \cite{GKKT2020} states that for every
$\rho\in\Theta_d$ and $\varepsilon\in
[0,1]$,
\[
  \PP_\rho\!\left\{\norm{\widehat\rho_n-\rho}_1\ge\varepsilon\right\}
  \le d\exp\!\left(-\frac{n\varepsilon^2}{86\,d\,r^2}\right),
  \qquad r=\min\{\rank\rho,\rank\widehat\rho_n\}\le d .
\]
For the Haar covariant measurement, Theorem~2 of \cite{GKKT2020} gives
$\PP_\rho\{\norm{\widehat\rho_n-\rho}_1\ge\varepsilon\}\le\exp(2.2d-
n\varepsilon^2/(480r^2))$ for every $\varepsilon>0$. After bounding
$r\le d$, both cases have the form
\begin{equation}
  \PP_\rho\!\left\{\norm{\widehat\rho_n-\rho}_1\ge\varepsilon\right\}
  \le\exp\!\left(\alpha-\frac{n\varepsilon^2}{\beta}\right),
  \qquad\varepsilon\in(0,1],
  \label{supp-eq:tail}
\end{equation}
with $(\alpha,\beta)=(\log d,\,86d^3)$ for finite designs and
$(\alpha,\beta)=(2.2d,\,480d^2)$ for the Haar covariant measurement.

\begin{lemma}
\label{supp-lem:tomography-mse}
Let $X\in[0,2]$ satisfy $\PP(X\ge\varepsilon)\le\exp(\alpha-n
\varepsilon^2/\beta)$ for $\varepsilon\in(0,1]$, with $\alpha,\beta>0$
and $n\ge2\alpha\beta$. Then $\E X^2\le\beta(\alpha+7)/n$.
\end{lemma}

\begin{proof}
Put $\varepsilon_0^2=\alpha\beta/n\le1/2$. Then
\begin{align*}
  \E X^2
  &=\int_0^2 2\varepsilon\,\PP(X\ge\varepsilon)\,\dd\varepsilon
  \le\varepsilon_0^2
  +\int_{\varepsilon_0}^1 2\varepsilon
  \exp\!\left\{-\frac{n(\varepsilon^2-\varepsilon_0^2)}{\beta}\right\}
  \dd\varepsilon
  +3\,\PP(X\ge1)\\
  &\le\frac{\alpha\beta}{n}+\frac\beta n
  +3\exp\!\left\{-\frac{n(1-\varepsilon_0^2)}{\beta}\right\}
  \le\frac{\beta(\alpha+1)}n+3\exp\!\left(-\frac n{2\beta}\right)
  \le\frac{\beta(\alpha+7)}n ,
\end{align*}
using $e^{-x}\le1/x$ for $x>0$ in the last step.
\end{proof}

For $d\ge2$ one has $\log d\ge\log2$. In both cases of
\eqref{supp-eq:tail}, the hypothesis $n\ge2\alpha\beta$ therefore holds
whenever $n\ge3100\,d^3\log d$, and
Lemma~\ref{supp-lem:tomography-mse} gives
\begin{equation}
  \E_\rho\norm{\widehat\rho_n-\rho}_1^2
  \le\frac{4000\,d^3\log d}{n},
  \qquad\rho\in\Theta_d,\quad n\ge3100\,d^3\log d .
  \label{supp-eq:tomography-mse}
\end{equation}
For the finite-design bound,
$86(\log d+7)\le86\{1+7/\log2\}\log d\le955\log d$. For the Haar bound,
$480d^2(2.2d+7)\le2736d^3\le3948d^3\log d$.

\begin{lemma}
\label{supp-lem:continuity}
Let $d\ge2$, $\rho,\sigma\in\Theta_d$, and $T=\frac12\norm{\rho-\sigma}_1$.
Then $|S(\rho)-S(\sigma)|\le T\log d+2\sqrt T$.
\end{lemma}

\begin{proof}
The sharp continuity estimate of Audenaert \cite[Theorem~1]{Audenaert2007}
states that $|S(\rho)-S(\sigma)|\le T\log(d-1)+h(T)$ for all
$\rho,\sigma\in\Theta_d$, where $h(T)=-T\log T-(1-T)\log(1-T)$; the
statement is given there in base-two logarithms and is invariant under a
change of base. Writing $T=u^2$, $-T\log T=2u(-u\log u)\le2u/e$, and
$-(1-T)\log(1-T)\le T\le\sqrt T$ for $T\in[0,1]$, so
$h(T)\le(2/e+1)\sqrt T\le2\sqrt T$.
\end{proof}

\begin{proof}[Proof of Proposition~\ref{prop:upper}]
Let $n\ge3100\,d^3\log d$ and $T_n=\frac12\norm{\widehat\rho_n-\rho}_1$.
By Lemma~\ref{supp-lem:continuity} and $(x+y)^2\le2x^2+2y^2$,
\[
  \{\widehat S_n-S(\rho)\}^2\le2T_n^2\log^2d+8T_n ,
\]
and by \eqref{supp-eq:tomography-mse}, $\E_\rho T_n^2\le1000\,d^3\log
d/n$. Jensen's inequality gives $\E_\rho T_n\le(\E_\rho T_n^2)^{1/2}$,
so
\[
  \E_\rho\{\widehat S_n-S(\rho)\}^2
  \le\frac{2000\,d^3\log^3d}{n}
  +8\left(\frac{1000\,d^3\log d}{n}\right)^{1/2}
  \le2000\left\{\frac{d^3\log^3d}{n}
  +\left(\frac{d^3\log d}{n}\right)^{1/2}\right\}
\]
uniformly in $\rho$. This proves the proposition with $C=3100$ and
$C'=2000$.
\end{proof}

\section{Additional cases for the measurement model}
\label{supp-sec:scope}

\paragraph{Shared randomness across copies}
Let $\nu$ be an exact $2$-design, let $d$ be even, and take $r=d/2$ in
Lemma~\ref{supp-lem:kernel}. Then $\kappa_\nu=1$ and
\[
  1+\chi^2\!\left(P_{\pi_{d/2}}^{(n)}\,\middle\|\,P_{I/d}^{(n)}\right)
  \le\exp\!\left(\frac{96n^2}{d^4}\right),
\]
so the total variation distance between the Haar mixture of flat
rank-$d/2$ states and the null law tends to zero whenever $n=o(d^2)$. In
particular, no test based on independent outcomes of an exact $2$-design
uniformly separates $I/d$ from this alternative when $n=o(d^2)$. Every state
in the alternative satisfies
$\norm{\rho-I/d}_1=1$, and the mixedness tester of
\cite[Appendix~A]{BCL2020}, which draws one Haar-random orthonormal
basis, measures every copy in that basis, and applies a classical
uniformity test to the outcome frequencies, distinguishes $I/d$ from
any state with $\norm{\rho-I/d}_1\ge1$ using $O(d^{3/2})$ copies. The
second protocol is nonadaptive and uses single-copy measurements, but its
outcomes share the random basis across copies. It therefore lies outside the
independent fixed-POVM experiment of the main theorem.

\paragraph{A single orthonormal basis}
Let $\nu$ be uniform on the projectors $\ket b\bra b$ of one
orthonormal basis. Then $\Fr(A)=\sum_b\langle b|A|b\rangle\ket b\bra b$
is the pinching onto the diagonal, which fixes every traceless diagonal
matrix, so $\norm{\Frz}_{2\to2}=1$ and $\kappa_\nu=d+1$.
Theorem~\ref{thm:main} then yields a lower-bound scale of order $d$
only.
The experiment is nonidentifying. A pure state
$\ket\psi\bra\psi$ with $|\langle b|\psi\rangle|^2=1/d$ for all $b$
produces the same uniform outcome law as $I/d$, while the corresponding
entropies are $0$ and $\log d$. Hence every estimator has worst-case squared
error at least $\frac14\log^2d$ for every $n$, and the frame-spread bound is
conservative in this case.

\paragraph{Finite POVMs with higher-rank effects}
Let $\{E_j\}_{j=1}^m$ be a POVM with $\sum_jE_j=I$, and write each
effect in spectral form $E_j=\sum_\ell e_{j\ell}P_{j\ell}$ with
$e_{j\ell}>0$ and rank-one projectors $P_{j\ell}$. Then
$\nu(j,\ell)=e_{j\ell}/d$ is a probability measure with
$\sum_{j,\ell}\nu(j,\ell)P_{j\ell}=I/d$, and the refined POVM
$\{e_{j\ell}P_{j\ell}\}$ is of the form considered in the main article.
The outcome $j$ of the original POVM is the image of the refined outcome
$(j,\ell)$ under a fixed map, so an estimator based on the original data can
be applied to the refined data after discarding $\ell$. Consequently the
minimax risk of the original experiment is at least that of the refined one.
The lower bound in Theorem~\ref{thm:main} therefore transfers to the
original POVM, with the frame spread computed for the chosen refinement.

\bibliographystyle{elsarticle-num}
\bibliography{refs}

\end{document}